\documentclass[
  a4paper, 
  reqno, 
  oneside, 
  11pt
]{amsart}

\usepackage[utf8]{inputenc}
\usepackage{soul}

\usepackage[dvipsnames]{xcolor} 
\colorlet{cite}{LimeGreen!50!Green}
\usepackage{tikz}  
\usetikzlibrary{arrows,positioning}
\tikzset{ 
  baseline=-2.3pt,
  text height=1.5ex, text depth=0.25ex,
  >=stealth,
  node distance=2cm,
  mid/.style={fill=white,inner sep=2.5pt},
}
\usepackage{lmodern}
\usepackage{tikz-cd}
\usepackage{amsthm, amssymb, amsfonts}
\usepackage{soul}
\usepackage{cancel}
\input xypic

\usepackage{graphicx,caption,subcaption}
\usepackage{fourier}
\usepackage{braket}  

\usepackage[%
  bookmarks=true,			
  unicode=true,			
  pdftitle={},		%
  pdfauthor={}%
  pdfkeywords={},	
  colorlinks=true,		
  linkcolor=Blue,			
  citecolor=cite,		
  filecolor=magenta,		
  urlcolor=RoyalBlue			
]{hyperref}				
\usepackage{cleveref}

\DeclareMathOperator{\Tot}{Tot}
\DeclareMathOperator{\Tr}{Tr}

\newcommand{\old}[1]{}

\newcommand{\ce}{\mathrel{\mathop:}=}

\theoremstyle{plain}	
\newtheorem{theorem}{Theorem}[section] 
\newtheorem{lemma}[theorem]{Lemma} 

\newtheorem*{theorem*}{Theorem}
\newtheorem{corollary}[theorem]{Corollary}

\theoremstyle{definition}
\newtheorem{example}[theorem]{Example}
\newtheorem*{conjecture*}{Conjecture}
\theoremstyle{remark}
\newtheorem{remark}[theorem]{Remark}
\newtheorem{definition}[theorem]{Definition}

\newtheorem*{lemma*}{Lemma}

\usepackage{amsmath,amssymb, xcolor}
\author{E. Ballico}
\address{EB: Dept. Mathematics, University of Trento,   I-38050, Povo, Italy.}
\author{E. Gasparim}
\address{EG: Dept. Mathematics, Univ. Catolica del Norte, Antofagasta, Chile}
 \author{M. P. Garcia del Moral}
 \address{MPGM: Universidad de la Rioja \\ Área de Física, Departamento de Química, Centro Científico Tecnológico, Universidad de la Rioja, La Rioja, 26006 Spain}
\author{C. Las Heras}
\address{CLH: Instituto de Física Teórica UAM/CSIC, C/ Nicolás Cabrera 13-15 Universidad Aut\'onoma de Madrid 28040, Spain}
\thanks{}
\title[]{F-theory with hyperelliptic fibrations}
\begin{document}

\begin{abstract}
We discuss the role of hyperelliptic fibrations in F-theory. 
For each  \\
even integer $n$ we give a noncompact Calabi--Yau threefold $X$ 
containing a hyperelliptically fibered surface $Y$, such that $X$ and $Y$ are homotopy equivalent and $c_2(X) = n.$ 
We investigate two  distinct cases depending on the position of the hyperelliptic fibration. 
First, we propose to extend   F-theory considering  hyperelliptic fibrations,
giving an identification between  the determinant of the period matrix 
and the axio-dilaton. Such an identification requires that the curve satisfies  an appropriate criterium which we describe.
Our explicit examples  have split Jacobian, 
preserve the same number of degrees of freedom of usual F-theory, while 
allowing for the appearance of a greater variety of singularities. Second, when 
the hyperelliptic fibration is contained in the base of a Calabi--Yau fourfold,
 we 
 show that  tadpole cancellation conditions are satisfied for arbitrarily large values of $c_2(X)$.

\end{abstract}

\maketitle
  \tableofcontents

We recall that a hyperelliptic curve of genus $g\ge 2$ is a
 connected and compact curve $C$ such that there is a degree $2$ morphism $u\colon C\to \mathbb P^1$. All genus $2$ curves are hyperelliptic because the canonical system $|K_C|$ has target $\mathbb P^{g-1} =\mathbb P^1$ and degree $2g-2=2$. When working in characteristic different from 2, certainly the case here, as we consider complex curves,
 a hyperelliptic curve of genus $g$ can be written in the form  $$y^2= f(x)$$ 
 where $f$ is a polynomial with distinct roots of degree either $2g+2$ or $2g+1$.


\section{Introduction}

We analyze the role of hyperelliptic fibrations in F-theory;
in particular,  fibrations by curves of genus 2. In this work, each hyperelliptically fibered surface 
 is consistently embedded inside a 
noncompact Calabi--Yau threefold ($CY_3)$ which is itself  embedded into a Calabi--Yau fourfold $(CY_4)$, giving in 
either case a total of 6 compact dimensions. 
 Two  very different cases are considered: 
 \begin{itemize}
 \item
  the first case,  named  {\it upstairs}, where the $CY_3$ is placed in the position of fiber of the $CY_4$:
  in which case the hyperelliptic curve itself may also 
 be regarded as living in  fibers of the $CY_4$ (by composition the 2 fibrations);

 \item
  and the second case, named {\it downstairs},
 where the $CY_3$  is positioned at the base of an elliptically fibered $CY_4$ as  standard in F-theory, and
 where we explore some the additional  features brought up by the presence of the hyperelliptic  fibration inside the $CY_3$.
 \end{itemize}

F-theory compactifications are type IIB string theory compactifications where the backreaction of the 7-branes is considered \cite{Vafa}.
 It allows control of the nonperturbative limit by encoding the physics in algebro-geometric language; see \cite{Weigand} for a review. 
 The backreaction of the 7-branes generates a holomorphically varying profile of the axio-dilaton: 
\begin{equation}
 \lambda=C_0+i\exp{(-\phi)}, \label{Axio-dilaton}
 \end{equation}
 where $g=\exp{(\phi)}$ is the string coupling constant. The $SL(2,\mathbb{Z})$ invariance 
 appearing in type IIB string theory leads to a natural interpretation in terms of elliptic fibrations. The axio-dilaton is geometrized in terms of the complex 
structure moduli of elliptic curves, with both transforming by $SL(2,\mathbb{Z})$, see for example \cite{DW}. In principle, this choice does not exclude the possibility of other geometrical realizations.

In this work, we explore different scenarios where hyperelliptic curves play an important role in F-theory. We 
 consider hyperelliptic curves of genus 2 for simplicity. Notice that the generalization of the elliptic fibration to
  a hyperelliptic one is not at all straightforward when we want to preserve the property of the total space being Calabi--Yau. 
  We bring forward an original yet canonical strategy to solve this problem, obtained by embedding the hyperelliptic fibered surface $Y$ into the 
  single noncompact $CY_3$
  which forms the total space of the canonical bundle of $Y$.\\

{\sc Case I}: The hyperelliptic fibration occurs upstairs. 
   In this case, the hyperelliptic curves may be seen as occurring in the fibers of the $CY_3$. 
   We provide a new interpretation of the axio-dilaton in terms of the determinant of the period matrix parametrizing the hyperelliptic curve (Sec.\thinspace \ref{up}). 
   
   Our proposal differs from the ones considered in 
\cite{Braun,Candelas1,Candelas2,MMP}, where an auxiliary K3
  fibration is considered, resulting in  G-theory formulated in 14 dimensions instead of 12,
  which was used to describe type IIB non-perturbative flux vacua. When those
  authors consider the K3 surface with three moduli, the U-duality group becomes isomorphic to $Sp(4,\mathbb{Z})$, 
  thus coinciding with the mapping class group of a hyperelliptic curve. For this reason, the authors claim an interpretation in terms of a certain hyperelliptic fibration, considering as the axio-dilaton a single entry of the period matrix. In contrast, in our case, by considering the axio-dilaton as the determinant of the period matrix, we obtain that varying any  
  of the entries of the period matrix influences the behaviour of the axio-dilaton.\\

{\sc Case II}: The hyperelliptic fibration occurs downstairs. 
In this case, the hyperelliptic fibration
 is contained on the base of an elliptically fibered $CY_4$. We briefly discuss the issue of tadpole cancellation and its relation with numerical invariants of the $CY_3$ (Sec.\thinspace  \ref{down}).    \\

  This work is organized as follows: in Section \ref{fib}, we recall  algebro-geometric results about fibered surfaces and construct families of curves containing any number of nodal fibers (Theorem \ref{aa3}) which in addition may contain any chosen nodal curve  (Corollary \ref{chosen}). 
  We discuss 
the relations between curves having split Jacobians and singularities of multiplicities at least $2$ or $3$ (Theorem \ref{23}),
which is of interest given that the specific  examples we present in Section  \ref{phcond} do satisfy the condition of having split Jacobians.

  In Section \ref{totomega}, we embed any given hyperelliptically fibered surface $Y$   into its corresponding noncompact Calabi--Yau threefold 
  $X= \Tot (\omega_Y)$, showing that all even integers may be obtained as $c_2(X)$ (Theorems \ref{ue1}, \ref{uet}).
Note that here $Y$ is a deformation retract of $X$, thus $H^4(X,\mathbb{Z})=H^4(Y,\mathbb{Z})=\mathbb{Z}$. In Section \ref{up}, we consider Case I, where we propose a new interpretation of the axio-dilaton,  as the determinant of the period matrices associated with the hyperelliptic fibration. Since the axio-dilaton must transform by  $SL(2, \mathbb Z)$, 
   we impose this condition on the determinant, which then forces the set of admissible hyperelliptic curves in our 
   fibrations  to 
   satisfy some explicit strict conditions (\ref{ph}). In Section \ref{phcond}  we give several examples, and in \ref{monodromies} we discuss some physical implications. In Section \ref{down}, we consider Case II, where
     we briefly discuss some phenomenological implications that the presence of hyperelliptic fibrations might bring about for
      the cancellation of tadpoles, and the matter content of the theory. Finally, in Section \ref{conclusion}, we discuss our results.\\

\section{Fibered surfaces }\label{fib}
We recall some algebro-geometric results about fibered surfaces
 that will be used  to determine the existence of new constraints 
 to the matter content in a generalisation of  F-theory formulated using a hyperelliptic fibration. 

In standard algebraic geometric notation,  $\mathcal O_X$ denotes the structure sheaf of $X$,
$\Omega_X^1$ the sheaf of 1-forms and 
$\omega_X$ the dualizing sheaf, as in \cite{Ha}, and $h^i(X)$ denotes the dimension of $H^i(X)$. 
For the fibered surfaces, we denote by $g_1$ the genus of the fiber and $g_2$ the genus of the base. 

\begin{definition}\label{chi}
Let $X$ be a smooth compact complex surface.
The {\bf holomorphic Euler characteristic} of $X$ is 
$$\chi_{hol}(X)\ce \chi (\mathcal O_X)= 1-h^1(\mathcal O_X) +p_g(X) = h^0(\mathcal O_X) -h^1(\mathcal O_X) +h^2(\mathcal O_X).$$
The (topological) {\bf Euler characteristic} of $X$ is
$$\chi(X) \ce \chi(TX) = h^0(TX) -h^1(TX) +h^2(TX)-h^3(TX) +h^4(TX) $$
(also for smooth compact manifold of real dimension 4). 
\end{definition}
 Noether's formula gives 
 $$c_1^2(X)+c_2(X)= (K\cdot K) +\chi(X)$$ 
  where $c_1^2$ and $c_2$ denote the Chern numbers and $K$ the canonical class of $X$.
  Note that if $X'$ is the surface obtained from $X$ by blowing up of one point, then 
  we have ve $c_1^2(X') =c_1^2(X)-1$ and hence $c_2(X') =c_2(X)+1$ \cite{BHPV}.

Let $f\colon X\to D$ be a proper holomorphic map with $X$ a smooth and connected complex surface (often called a fibration of curves with $D$ as a base) and $D$ a (not necessarily compact) Riemann surface (even not compact algebraic, e.g. a disc of $\mathbb C$). 

Assume $f_\ast (\mathcal O_X)=\mathcal O_D$, i.e. assume that a general fiber is connected (and so all fibers of $f$ are connected). Let $A$ be any smooth fiber of $f$ and $A_s\ce f^{-1}(s)$ any fiber of $f$. 
Then the Euler numbers satisfy 
$\chi(A_s)\ge \chi(A)$. If $X$ is compact, i.e. if $D$ is compact, then  by \cite[Prop.\thinspace III.11.4]{BHPV} we have

\begin{equation}\label{eqa1}
\chi(X) =\chi(A)\chi(D) +\sum _{s\in D} (\chi(A) -\chi(A_s).
\end{equation}
Moreover, there is an easy criterion which gives $\chi(A_s)>\chi(A)$ if $A_s$ is singular and not a multiple of a smooth elliptic curve \cite[III.11.5]{BHPV}. Hence if $D$ is compact of genus $g_2$ and a general fiber of $f$ is smooth of genus $g_1$, then 
$$\chi(X)\ge 4(g_1-1)(g_2-1)$$ 
with strict inequality if some of the fibers are singular and either $g_1\ne 1$ or $g_1=1$, but there is at least one fiber with is not a multiple of an elliptic curve \cite[III.\thinspace 11.6]{BHPV}.

Since any smooth genus 2 curve is hyperelliptic, we have that 
the simplest examples of hyperelliptic fibrations 
are those surfaces fibered by  Riemann surfaces of genus $2$.

 \subsection{Existence of surfaces with nodal fibers}
 
 If $L$ is a line bundle over a projective variety, the linear system corresponding to $L$ is denoted by $|L|$.
We have $ |L|\ce \mathbb{P} (H^0(L))$ 
 formed by the set of non-zero sections of $L$, up to a non-zero multiplicative scalar.

  The linear system $|L|$ can 
 equivalently be described as the family $|D|$ of divisors linearly equivalent to $D$. 
 For a projective surface $Y$,  effective divisors  correspond to curves in $Y$. 
 Therefore
 saying that  $\dim |L| = k$ this means that $|D|$ may be regarded as  a family of curves parametrized by $\mathbb P^k$.

 \begin{theorem}\label{aa3}
 Let $Y$ be a smooth projective surface. Fix an integer $k>0$. There exists a very ample line bundle $L$ on $Y$ with the following properties.
For each integer $t\in \{0,\dots ,k\}$ let $V(L,t)$ denote the set of all $C\in |L|$ which are integral, nodal and with exactly $t$ nodes. Then $V(L,t)\ne \emptyset$ and $\dim V(L,t) =\dim |L| -t$
for each $t\in \{0,\dots ,k\}$. Moreover, if $t>0$, then every element in $V(L,t)$ is in the closure of $V(L,t-1)$.
 \end{theorem}
 
 \begin{proof}
 For any $p\in Y$ let $2p$ (resp. $3p$) denote the closed subscheme of $Y$ with $(\mathcal{I}_p)^2$ (resp. $(\mathcal{I}_p)^3$) as its ideal sheaf. The schemes $2p$ and $3p$ are zero-dimensional, $p$ is their reduction, $\deg (2p) =3$ and $\deg(3p) =6$. For any finite set $S\subset Y$ set $2S:= \cup _{p\in S} 2p$ and $3S:= \cup _{p\in S} 3p$. Fix a very ample line bundle $R$ on $Y$ such that $h^i(R) =h^i(R\otimes \omega _Y^\vee) =0$ for all $i>0$. Take an integer $m\ge 6k$ and set $L:= R^{\otimes m}$. It is straightforward to prove that
 $h^1(\mathcal{I}_{3S}\otimes L) =0$ for all $S\subset Y$ such that $\#S\le k$. Note that $C\in |\mathcal{I}_{2p}\otimes L|$ if and only if $C$ is singular at  $p$ and that $\dim |\mathcal{I}_{3p}\otimes L| =\dim |L| -6$ implies that a general $C\in  |\mathcal{I}_{2p}\otimes L|$ is nodal at $p$. Since $m\ge 3k$ we easily see that $h^1(\mathcal{I}_{3S}\otimes L) =0$ for all $S\subset Y$ such that $\#S\le k$.
 Then a straightforward (but long) argument concludes the proof as in \cite{dedse,t}.
 \end{proof}

\begin{example}
  Consider the case of  $\mathbb{P}^2$. For any integer $d\ge 2$ and $t\ge 0$ let $V(\mathbb{P}^2,d,t)$ denote the set of all irreducible and nodal plane curves with exactly $t$ nodes.
  We have $V(\mathbb{P}^2,d,t) =\emptyset$ if $t>(d-1)(d-2)/2$, while each $V(\mathbb{P}^2,d,t)$ is irreducible and of dimension $(d^2+3d)/2-t$ for all $0\le t\le (d-1)(d-2)/2$. Moreover
  if $t>0$, then $V(\mathbb{P}^2,d,t)$ is contained in the closure of $V(\mathbb{P}^2,d,t-1)$, see \cite{harris2}.
  
  For many surfaces which are ``almost'' $\mathbb{P}^2$ one can do  the same \cite{ac,dedse,t,ty}. 
\end{example}
    
      \begin{corollary} \label{chosen} Any curve with at most nodal singularities belongs to a fibered surface over $\mathbb P^1$.
    \end{corollary}

   \begin{proof}    To construct a fibrations over $\mathbb P^1$ having  a chosen curve $C$ as a fiber,
  take $(Y,L)$ as in Proposition \ref{aa3} and set $\alpha:= L\cdot L$ (the self-intersection number). Fix $C\in V(L,k)$ and call $s\in H^0(L)$ an equation of $C$.
  Let $V\subset H^0(L)$ be a general $2$-dimensional linear subspace containing $s$. Note that $V$ induces a non-constant rational map $\psi$ from $Y$ to $\mathbb{P}^1$. The generality
  of $V$ means that the base locus of $V$ is given by $\alpha$ distinct points, none of them being contained in $C$. Let $Y_1$ be the blowing-up of $Y$ at these $\alpha$ points. The rational map $\psi$ induces a fibration
   $f\colon Y_1\to \mathbb{P}^1$ with $C$ as one of its fibers, while the general fiber of $f$ is a smooth curve. 
 \end{proof}

\begin{remark}\label{Rm3}In conclusion, we see that we may choose our family to have any given number of nodal fibers, and 
we may also choose the family to contain a specific chosen curve $C$ as a fiber. We could also construct fibrations with 
singularities with higher multiplicities, but we leave this for future work. 
\end{remark}

\subsection{Split Jacobians  and singularities}\label{SJ}
Here we mention a construction  connecting some hyperelliptic curves to 
elliptic ones. It applies to the examples \ref{ex1}-- \ref{ex5}  in section \ref{phcond}
Let C be a genus 2 curve  over $\mathbb C$ and let
 $J(C)$ be its Jacobian. An abelian variety  is called decomposable  if it  is isogenous  to a product of elliptic curves $E_1 \times E_2$.
The curve $C$ has a decomposable Jacobian if and only if there is a cover $\phi\colon  C \rightarrow  E_1$ to an elliptic curve $E_1$.

\begin{theorem}\label{23} Let $C$ be a genus $2$ smooth curve with split Jacobian. Then there exists an elliptic curve $E$ 
and a surjection $f\colon C\to E$ having ramification points which in local coordinates are given by either 
$z\mapsto z^2$ or $z\mapsto z^3$.
\end{theorem}

\begin{proof}
Let $C$ be a genus $2$ smooth curve such that there is an elliptic curve $E$ and a surjection $J(C)\to E$. The inclusion of $C$ in $J(C)$ induces
a surjection $f\colon C\to E$. Set $n:= \deg(f)>1$. Since $K_E\cong \mathcal O_E$ and $\deg(K_C)=2$, 
the Riemann--Hurwitz formula gives that the ramification divisor has degree $2$. Hence
one of the following $3$ cases occur:
\begin{enumerate}
\item there are $p,q\in E$ such that $p\ne q$, $\#f^{-1}(o) =n$ for all $o\in E\setminus \{p,q\}$, $\#f^{-1}(p) =\#f^{-1}(q) =n-1$. 
\item $n\ge 3$ and there is $p\in E$ such that $\#f^{-1}(o) =n$ for all $o\in E\setminus \{p\}$, $\#f^{-1}(p) = n-2$ and over $p$ there is a unique $a_p\in C$ at which $f$ ramifies.
\item $n\ge 4$ and there is $p\in E$ such that $\#f^{-1}(o) =n$ for all $o\in E\setminus \{p\}$, $\#f^{-1}(p) = n-2$ and over $p$ there are two distinct points $o_1,o_2\in C$ at which $f$ ramifies.
\end{enumerate}
In case (1) at the unique ramification point of $C$ over $p$ and the unique ramification point of $C$ over $q$ in local coordinates $f$ is the map $z\mapsto z^2$.
In case (2) at the unique ramification point of $C$ over $p$ in local coordinates $f$ is the map $z\mapsto z^3$.
In case (3) at each of the $2$ ramification points of $C$ over $p$ in local coordinates $f$ is the map $z\mapsto z^2$.
\end{proof}

\begin{remark}By Hurwitz existence theorem, for each elliptic curve $E$, $n$, 
and choice of  ramification points of $E$ ($p$ and $q$) in case (1), $p$ in cases (2) and (3),
 there, up to isomorphisms of pairs
and isomorphisms of $E$ at least one and at most finitely many pairs $(C,f)$. 
Now assume that we start with the elliptic curve  $(E,O_E)$ in case (1).
 Assuming that $p=O_E$, then for $q$ we have a $1$-dimensional family of possible $q$'s. 
 In cases (2) and (3) we may take $p=O_E$ and we do not have another parameter for the pairs $(C,f)$.
Deforming $E$ within elliptic curves we get another $1$-dimensional parameter for curve $C$ with a degree $n$ covering, the $j$-invariant.
\end{remark}

\begin{remark}
Assume that the degree $n$ covering $f\colon C\to E$ is induced by the action of a finite group $G$ of automorphisms of $C$, say $E =C/G$.
 In case (1) we have $n=2$ and hence $G\cong \mathbb Z/2\mathbb Z$. In case (2) we have $n=3$ and 
 $G\cong \mathbb Z/3\mathbb Z$. In case (3) we have $n=4$, the covering is 2 to 1 covering $u\colon C\to E_1$ with $E_1$ an elliptic curve followed
by a 2 to 1 covering $E_1\to E$. Hence (assuming that it comes from a quotient $E =C/G$), if the covering is induced by a the action of 
a finite group, case (3) gives nothing new.
\end{remark}

\section{Embedding the fibered surface in a $CY_3$}\label{totomega}

  In this geometric section we show existence of fibered surfaces containing a prescribed number 
of nodal fibers, discuss  hyperelliptic families with nodes, and embed fibered surfaces in Calabi--Yau threefolds.

\subsection{Constructing $CY_3$'s with prescribed $c_2$}

We use the adjunction formula. 
Let $X$ be a smooth algebraic variety or smooth complex manifold and $Y$ be a smooth subvariety of $X$.
 Denote by $i\colon Y \rightarrow X$ the inclusion map and  by $\mathcal I$ the ideal sheaf of $Y$ in $X$. 
 The conormal bundle of $Y$ in $X$ is ${\mathcal  {I}}/{\mathcal  {I}}^{2}$ and the conormal exact sequence for $i$ is
\begin{equation}\label{conormal} 0\to {\mathcal  {I}}/{\mathcal  {I}}^{2}\to i^{*}\Omega _{X}\to \Omega _{Y}\to 0, \end{equation}
where $\Omega$  denotes a cotangent bundle. The determinant of this exact sequence is a natural isomorphism
$$\omega _{Y}=i^{*}\omega _{X}\otimes \operatorname {det}({\mathcal  {I}}/{\mathcal  {I}}^{2})^{\vee }$$
where 
$\vee$  denotes the dual of a line bundle, and $\omega$ denotes de canonical bundle (in some references $\omega$ is  
 denoted by $K$).

\begin{example}\label{aaa1}
Now we take $X$ to be equal to the total space of the canonical bundle of a complex surface $Y$ (4 real dimensions),
 that is 
$X= \Tot(\omega_Y)$. We claim that $X$ is Calabi--Yau.
The fiber of $X \rightarrow Y$ at a point $q \in Y$ is $\mathbb C$ because 
$X$ is a line bundle over $Y$.

 Indeed, observe that the normal bundle of $Y$ inside $X$ will be exactly equal to $\omega_Y$, 
that is, 
$$({\mathcal  {I}}/{\mathcal  {I}}^{2})^\vee = \omega_Y$$
and
since the surface $Y$ is embedded in the threefold $X$ in codimension 1, we also have
$$\operatorname{det}({{\mathcal  {I}}/{\mathcal  {I}}^{2}})^\vee = \omega_Y.$$
Plugging this information into the previous formula, we obtain
$$\omega _{Y}=i^{*}\omega _{X}\otimes \omega_Y.$$
Therefore $i^{*}\omega _{X}$ must be trivial, and because it is a line
bundle whose restriction to the zero section  is trivial,
 we also obtain that $\omega _{X}$ is trivial. This shows that $X$ is a Calabi--Yau threefold. 
\end{example}

    \begin{remark}
  Take $(Y,L)$ as in Proposition \ref{aa3} and set $\alpha:= L\cdot L$ (the self-intersection number). Fix $C\in V(L,k)$ and call $s\in H^0(L)$ an equation of $C$.
  Let $V\subset H^0(L)$ be a general $2$-dimensional linear subspace containing $s$. Note that $V$ induces a non-constant rational map $\psi$ from $Y$ to $\mathbb{P}^1$. The generality
  of $V$ means that the base locus of $V$ is given by $\alpha$ distinct points, none of them being contained in $C$. Let $Y_1$ be the blowing-up of $Y$ at these $\alpha$ points. The rational map $\psi$ induces a fibration $f: Y_1\to \mathbb{P}^1$ with $C$ as one of its fibers, while the general fiber of $f$ is a smooth curve. By Example \ref{aaa1} $\Tot(\omega_{Y_1})$ has trivial canonical bundle.
  \end{remark}

\begin{example}\label{aa2} 
Continuing from the previous example, assume in addition that
$Y$ is a fibered surface. Hence, we have that $Y$ is the total space of a fibration
$Y\rightarrow C $ for some curve $C$. Then for each point $p\in C$ the fiber of $Y$
is a curve $F_p$.
We have  that the fiber over $p$ of the fibrations $\pi\colon X \rightarrow C$
(obtained from the composition of the projections   $X\rightarrow Y$ and $ Y\rightarrow C $) is
$\omega_{F_p},$ that is,  the canonical bundle of $X$ restricted to the fiber $F_p$ gives the canonical bundle
of the $F_p$.
Therefore, the fiber of $p$ is a trivial product $F_p\times \mathbb C$ if and only if $F_p$ is smooth of genus 1.
\end{example}

In the particular case when $Y$ is a locally trivial fibration,
we  write $F \rightarrow Y\rightarrow C $
 and for each $p \in C$ we have the existence of a disc neighbourhood
$p\in D $ such that $Y\vert_D \simeq D \times F$. In this case, then $X= \Tot(\omega_Y)$
may  also be seen as a locally trivial fibration. In fact,  
in principle we write $\mathbb C \rightarrow X \rightarrow Y$, but we
  may also write $\omega_F \rightarrow X\rightarrow C $.

 \begin{lemma}\label{c2}
 The Chern classes of $X = \Tot(\omega_Y)$ are 
 $c_1(X) = c_3(X) =0$ and $c_2(X)  =  -c_1^2(Y)+c_2(Y).$ 
 \end{lemma}
 
 \begin{proof} To  calculate the second Chern class of 
 $X= \Tot(\omega_Y)$  consider the exact sequence dual to \ref{conormal}:
$$ 0 \rightarrow TY\rightarrow i^{*}T X\rightarrow
 ({\mathcal  {I}}/{\mathcal  {I}}^{2})^{\vee } \rightarrow  0.$$
 Given that the Chern polynomial is multiplicative and $ ({\mathcal  {I}}/{\mathcal  {I}}^{2})^{\vee }=\omega_Y$, we obtain:
 $$c_t(TX) = c_t(\omega_Y) c_t(TY).$$
 Since $X$ retracts to the surface $Y$, we have that $c_3(TX) = 0$, giving
  $$1+c_1(TX)t+c_2(TX)t^2  = (1+c_1(\omega_Y)t)(1+ c_1(TY)t+c_2(TY)t^2).$$
  Since $c_1(\omega_Y)= -c_1(TY)$, we have
   $$c_1(TX)  = c_1(\omega_Y)+c_1(TY)= 0 ,$$
   which we already knew because $X$ is Calabi--Yau. Moreover
     $$c_2(TX)  = c_1(\omega_Y)c_1(TY)+c_2(TY)= -c_1^2(TY)+c_2(TY),$$ 
  or equivalently
  $$c_2(X)  =  -c_1^2(Y)+c_2(Y).$$ 
 \end{proof}
 
  In Example \ref{aaa1} we may take as $Y$ an arbitrary smooth surface. Example \ref{aa2} explained the case in which $Y$ is a fibered surface. 
  Here is the simplest  illustration of  proposition \ref{c2}.

\begin{example} Suppose $Y= \mathbb P^1 \times \mathbb P^1 $ and we call $B$ the base 
and $F$ the fiber. Then we have
$c_1(Y) = 2B+2F $, therefore 
$$c_1^2(Y)= (2B+2F)^2 = 4B^2+8B\cdot F + 4F^2 = 8$$
while $c_2(Y) = c_1(\mathbb P^1) \times c_1(\mathbb P^1) = 4$.
Now, taking the Calabi--Yau threefold $X= \Tot(\omega_Y)$ we obtain $c_1(X)=0=c_3(X)$ and 
$$c_2(X) = -c_1^2(Y)+c_2(Y) = - 8+4 = -4.$$
\end{example}

To state the next result, 
we will use the concept of holomorphic Euler characteristic $\chi(\mathcal O_Y)$, see 
definition \ref{chi}.

 For a smooth compact complex surface $X$ (i.e. $\dim_{\mathbb C}X=2$), Noether's formula gives 
 $$12\chi(\mathcal O_X) =c_1^2(X)+c_2(X)= (K\cdot K) +e(X)$$ 
 where $K$ is the canonical divisor class.
 In contrast, we also observe that if $C$ is a smooth compact curve (i.e. a Riemann surface) and
 hence the two numbers $h^0(\Omega ^1_C)$ and $h^1(\mathcal O_C)$ are the same, 
 we have that the topological and holomorphic Euler characteristics satisfy
 $2\chi(\mathcal O_C)=\chi(C) .$ 
  
  \begin{theorem}\label{ue1}
  If $Y$ a smooth projective surface and $X= \Tot(\omega_Y)$, then $c_2(X) $ is even.
  \end{theorem}
  \begin{proof}
For  any smooth and connected complex surface $Y$, we have 
$$
12\chi(\mathcal O_Y) = c_1(Y)^2 +c_2(Y),
$$
see \cite[Eq.(4) p.26]{BHPV} or \cite[p.472]{GH}.
Consequently, we have 
  $$-c_1(Y)^2 +c_2(Y) =12\chi(\mathcal O_Y)-2c_1(Y)^2.$$ Hence $$-c_1^2(Y)+c_2(Y)\equiv 0\mod{2}.$$
  \end{proof}

    \begin{theorem} \label{rat}
  Every even integer $\geq -6$ occurs as $c_2(X)$ for some rational 
  Calabi--Yau threefold of the form  $X = \Tot(\omega_Y)$.
  \end{theorem}
  \begin{proof}
  We may use either the cases starting from $\mathbb P^2$ as in example \ref{ue2}
 or else the examples starting from a Hirzebruch surface as in example \ref{ue3} below. 
  \end{proof}

  \begin{example}\label{ue2}
  Take any smooth and connected complex surface $W$ and let $Y$ be obtained from $W$ making $x$ blowing ups.
  We have $c_1(Y)^2 =c_1(W)^2+x$ and $c_2(Y) =c_2(W) +x$, therefore
  \begin{equation}\label{blowups} -c_1^2(Y)+c_2(Y) =-c_1^2(W)+c_2(W) +2x.\end{equation}
  Since $c_1(\mathbb P^2)^2=9$ and $c_2(\mathbb P^2) = 3$, we have $-c_1(\mathbb P^2)^2+c_2(\mathbb P^2) =-6$. Making some blowing ups of $\mathbb P^2$ we get
  that {\bf every even integer $\ge -6$} occurs as $-c_1^2(Y)+c_2(Y)$ for some smooth rational surface.
  \end{example}
  
  \begin{example}\label{ue3}
  Let $Y$ be the Hirzebruch surface $F_e$, $e\ge 0$. Since $Y$ is smooth and rational, $\chi(\mathcal O_Y)=1$. We take a fiber $f$ of a ruling of $F_e$ and a section $h$ of its ruling with $h^2=-e$ and a $\mathbb Z$-basis of $\mathrm(Y)
  =\mathbb Z h+\mathbb Z f$. We have $f^2=0$ and $h\cdot f=1$. The adjunction formula gives 
  $\omega_Y\cong \mathcal O_Y(-2h-(2+e)f)$. Thus $c_1(Y)^2 =-4e +8+4e =8$.
  Every smooth rational surface is obtained either from $\mathbb P^2$ or from a Hirzebruch surface $F_e$, $e>0$, making some blowing ups. By example \ref{ue2} the integers
  appearing as $-c_1^2(Y)+c_2(Y)$ for some smooth rational surface are exactly the even integers $\ge -6$.
  \end{example}
  
  \begin{theorem}\label{uet}
  Every even integer occurs as $c_2(X)$ for some Calabi--Yau threefold of the form  $X = \Tot(\omega_Y)$.
  \end{theorem}
  
  \begin{proof}\label{ue4}
  Recall that $c_1(Y)^2\le 3c_2(Y)$ for any smooth surface of general type $Y$ and that equality holds exactly for a class of surfaces (ball quotients) which exist, but
  they are difficult to construct \cite{Ishi,Ishi0}. If $Y$ is a minimal surface of general type, then $c_1(Y)^2>0$. Thus, 
  when $Y$ is a minimal surface of general type and $2c_2(Y)\le c_1(Y)^2\le 3c_2(Y)$, 
  then the even integer $-c_1(Y)^2+c_2(Y)$ is negative,  and we obtain that its modulus goes to $+\infty$,  
  provided  we get a sequence of surfaces $Y_n$ such that $c_1(Y_n)^2\ge 2c_2(Y_n)$ and $\lim_n c_2(Y_n)=+\infty$. 
  Indeed, in such case
  $$\lim_n (c_1(Y_n)^2-c_2(Y_n))\geq \lim_n (2c_2(Y_n)-c_2(Y_n)) \geq \lim_n c_2(Y_n)= \infty.$$ 
  Such a sequence of surfaces (even with additional properties) is constructed in \cite{RU,Som}.  
  So,  we see that $c_2(X)$ can be made arbitrarily low by using minimal surfaces $Y$ of general type, and
   blowing up points of $Y$  (if needed to fill up gaps), using
    formula \eqref{blowups}, we verify that all negative integers can occur as 
  $c_2(X)$ using (possibly nonminimal) surfaces of general type.
  Finally,   combining with Theorems \ref{ue1} and \ref{rat}, 
  we conclude that 
 {\bf every even integer} is of the form $-c_1(Y)^2+c_2(Y)$ for some smooth projective surface.
 \end{proof}

  So, we conclude that our collection of  Calabi--Yau threefolds $X$  contains representatives of  many types of 
  compact 4-cycles.
We end this section with  some comments about bounds on numerical invariants.  
  If $Y\rightarrow C$ is a fibered surface with fiber genus $g_1$ and base genus $g_2$ we have 
  $\chi(\mathcal O_X) \ge 2(g_1-1)(g_2-1).$
see \cite[Sec.\thinspace 2.2]{BGS}. For the construction of families of hyperelliptic curves 
with prescribed number of nodal fibers, see  \cite[Sec.\thinspace 2.3]{BGS}, 
where 3 types of construction of such families were presented.

  \section{Applications to constructions in F-theory I}\label{up}

\subsection{The hyperelliptic fibration occurs upstairs}
  
Nonperturbative solutions of type IIB supergravity are characterized in F-theory by the nontrivial profile of the axio-dilaton 
$\lambda$ given in (\ref{Axio-dilaton}), undergoing  a $SL(2,\mathbb{R})$ transformation (and in high energies $SL(2,\mathbb{Z})$) 
  that implies a duality between strong and  weak coupling regimes. The axio-dilaton 
  may be regarded as  the complex structure parameter of an auxiliary torus, namely an elliptic curve. 
  Consequently, F-theory describes type IIB string theory in terms of complex elliptic fibrations. We know that F-theory on an elliptically fibered   $CY_2$ (e.g., a $K3$ surface) is equivalent to type IIB string theory compactified on the base of this elliptic fibration (e.g., a $\mathbb{P}^1$). When moving around the points of the base manifold where the fiber degenerates, $\lambda$ undergoes a nontrivial monodromy 
  \begin{eqnarray*}
M= \begin{pmatrix}
    1 + pq & p^2 \\
    -q^2 & 1-pq
\end{pmatrix}    
\end{eqnarray*}
  contained in the parabolic subgroup of the U-duality group $SL(2,\mathbb Z)$, indicating the presence of $(p,q)$ 7-branes. 
  
Here we will consider the {\it unconventional} case when F-theory is formulated on a noncompact $CY_3$ given by $X=\Tot(\omega_Y)$, 
where a $Y$ is a hyperelliptic fibration over $\mathbb{P}^1$ whose fiber at a point $p$, denoted by $\Sigma_p$, 
 may degenerate in certain points of the base manifold. From example \ref{aa2}, we have that 
\[
\begin{tikzcd}[swap]
	\Tot(\omega_{\Sigma_p}) \arrow[]{r}[] {}
	& X \arrow[]{d}[right]{} \\  
		& \mathbb{P}^1 
\end{tikzcd}
\]
where,  the fibers vary. However, outside the set of singular fibers, 
we have a locally trivial fibration.

 The interest of this study relies on the fact that it may be regarded a 
natural extension of F-theory.  Indeed, we could also consider
\[
\begin{tikzcd}[swap]
	\Tot(\omega_{\Sigma_p}) \arrow[]{r}[] {}
	& CY_4 \arrow[]{d}[right]{} \\  
		& \mathbb{P}^1 \times T^2
\end{tikzcd}
\]
where $CY_4=X\times T^2$.


\subsection{The axio-dilaton in  hyperelliptic F-theory}\label{AX}
To obtain a description of F-theory  containing a hyperelliptic fibration 
upstairs, it is first of all necessary to determine the realization  of the axio-dilaton in 
this proposed new formulation.
The axio-dilaton distinguishes itself from  other moduli parameters 
that appear in string compactifications by the fact that the imaginary 
part of its contribution determines the string coupling, and it is 
strictly positive. It transforms under S-duality as
 \begin{equation*}
 \lambda^{'} =\frac{a\lambda+b}{c\lambda+d}  
 \end{equation*}
 with 
 $\tiny
 \begin{pmatrix}
        a & b \\ c& d
    \end{pmatrix}\in SL(2,\mathbb{Z}).
$
Therefore, it is a fundamental field that must be identified. 
In F-theory, the axio-dilaton is interpreted geometrically in terms of an elliptic fibration.
Therefore, for a hyperelliptic fibration, such as a surface fibered by genus 2 curves, where more moduli 
parameters characterize the curves,
 the field associated with the axio-dilaton must be identified. 

Let us consider a hyperelliptic fibration of genus two. The period matrix $\Omega$ describes the variations of the complex scalar fields on the base manifold
 \begin{eqnarray}\label{Periodmatrix}
\Omega(z) &=& \begin{pmatrix}
    \tau & \beta  \\
    \beta & \sigma
\end{pmatrix}
\end{eqnarray}
with 
\begin{eqnarray*}
    \tau &=& \tau_1 + i\tau_2, \\
    \sigma &=& \sigma_1 + i\sigma_2, \\
    \beta &=& \beta_1 + i \beta_2.
\end{eqnarray*}
The fiber over a generic point $z_0$ of the base manifold can be interpreted as the direct sum of two genus-one tori. Then $\tau$ and $\sigma$
 are the Teichm\"uller parameters of the 2-torus and $\beta$ the complex scalar field associated with the sewing of the two tori. 
 Therefore, $\tau_2>0$ and $\sigma_2>0$. In general, $\tau$, $\sigma$, and $\beta$ (consequently, $\Omega$) will depend on the coordinates of the base manifold.
 
 Then, $\Omega$ transforms as
 \begin{eqnarray*}
 \Omega = \frac{A\Omega + B}{C\Omega + D}
 \end{eqnarray*}
 where $A, B, C$, and $D$ are $2\times 2$ matrices which define a matrix of $Sp(4,\mathbb{Z})$, the mapping class group of a genus two Riemann surface via
 \begin{eqnarray*}
M = \begin{pmatrix}
    A & B  \\
    C & D
\end{pmatrix}
\end{eqnarray*}
such that
\begin{eqnarray*}
      M\begin{pmatrix}
     0 & \mathbb{I}  \\
     -\mathbb{I} & 0
\end{pmatrix} M^T = \begin{pmatrix}
     0 & \mathbb{I}  \\
     -\mathbb{I} & 0
\end{pmatrix}. \label{Symplecticgroup}
\end{eqnarray*}

We find that a natural choice is to identify the axio-dilaton with the complex 
scalar associated with the determinant of the period matrix $\Omega$, 
which carries  the topologic/geometric information of the 
hyperelliptic fiber.  Since the period matrix transforms according to 
$Sp(4,\mathbb{Z})$, it becomes necessary to restrict the choice of allowed period 
matrices to those whose  determinant  transforms under $SL(2,\mathbb{Z})$. 
Therefore, imposing this 
condition implies on a restriction on the allowed families of hyperelliptic 
fibrations.  One may argue that considering  this interpretation of the axio-dilaton 
is a natural choice from the physics perspective, since the 
Teichm\"uller parameter of an elliptic fibration also corresponds to 
the determinant of a one-times-one period matrix associated with the 
genus 1 curve.

Hence for the hyperelliptic fibration, we propose that the 
axio-dilaton is given by $$\lambda=\mbox{det}(\Omega).$$

Therefore, the corresponding transformation on $\lambda$ which guarantees that it transforms with $SL(2,\mathbb Z)$ 
will be the one inherited from the period matrix, that is
\begin{eqnarray}
    \lambda'&=& \det{\left(\frac{A\Omega+B}{C\Omega+D}\right)}
    = \frac{a\lambda + b}{c\lambda+d}, \label{axio-dilaton trans}
\end{eqnarray}
such that $a,b,c$, and $d$ are entries of a $SL(2,\mathbb{Z})$ matrix given by
\begin{eqnarray}
    a &=& \mbox{det}(A), \label{a} \\
    c &=& \mbox{det}(C),\label{c}  \\
    b &=& \tau_1 (B_{22}A_{11}-B_{12}A_{21}) + \sigma_1(B_{11}A_{22}-B_{21}A_{12}) , \nonumber \\
    &+& \beta_1(B_{22}A_{12}+B_{11}A_{21} - B_{21}A_{11}-B_{12}A_{22}) + \det{B}, \label{b} \\
    d &=& \tau_1(D_{22}C_{11}-D_{12}C_{21}) + \sigma_1(D_{11}C_{22}-D_{21}C_{12}), \nonumber \\
    &+& \beta_1(D_{22}C_{12}+D_{11}C_{21} - D_{21}C_{11}-D_{12}C_{22}) + \det{D}, \label{d} 
\end{eqnarray}
with the following conditions on $\tau_2$ and $\sigma_2 $ that guaranty that the imaginary parts vanish and 
on   $\tau_1$ and $\sigma_1 $ such that 
the real parts are integers
\begin{eqnarray}
    \tau_1 &=&\frac{1}{k_2}\left[d'(B_{11}A_{22}-B_{21}A_{12})-b'(D_{11}C_{22}-D_{21}C_{12}) \right. 
     + \left. k_3 \beta_1 \right] \label{Retau} \\
     \tau_2 &=&  \frac{k_3}{k_2}\beta_2 \label{Imtau} \\
     \sigma_1 &=& \frac{1}{k_2}\left[b'(D_{22}C_{11}-D_{12}C_{21})-d'(B_{22}A_{11}-B_{12}A_{21}) \right. 
     + \left. k_1 \beta_1 \right] \label{Resigma} \\
     \sigma_2 &=& \frac{k_1}{k_2}\beta_2 \label{Imsigma}
\end{eqnarray}
where $b'=(b-\det{B})$,  $d'=(d-\det{D})$
and
\begin{eqnarray*}
    k_1 &=& (D_{11}C_{21}+D_{22}C_{12}-D_{21}C_{11}-D_{12}C_{22})(B_{22}A_{11}-B_{12}A_{21})
    \nonumber \\ &-& (B_{11}A_{21} +B_{22}A_{12}- B_{21}A_{11}-B_{12}A_{22})(D_{22}C_{11}-D_{12}C_{21}) ,  \\
    k_2 &=& (D_{22}C_{11}-D_{12}C_{21})(B_{11}A_{22} - B_{21}A_{12}) \nonumber \\
    &-& (B_{22}A_{11}-B_{12}A_{21})(D_{11}C_{22}-D_{21}C_{12}), \\
    k_3 &=& (B_{11}A_{21} +B_{22}A_{12}- B_{21}A_{11}-B_{12}A_{22})(D_{11}C_{22}-D_{21}C_{12}) \nonumber \\ &-&  (D_{11}C_{21}+D_{22}C_{12}-D_{21}C_{11}-D_{12}C_{22})(B_{11}A_{22}-B_{21}A_{12}).
\end{eqnarray*}
The integers $k_1,k_2$, and $k_3$ are written in terms of the entries 
of the symplectic matrix $M$. Let us also notice that as 
$\tau_2,\sigma_2 > 0$, then $\beta_2\neq 0$, meaning that we can 
consider the hyperelliptic curve in terms of the sewing of two tori. 
It can not be factorized as the product of two torus.

Therefore, we propose that the axio-dilaton is written in terms of the 
three complex scalars that characterize  as 
the determinant of the period matrix of the genus 2 curve. This is a different approach 
from the one used in \cite{MMP}, where the axio-dilaton is identified 
with the single  entry  $\tau$ of the period matrix.  However, it is important to 
notice that the number of degrees of freedom of our extension of   
F-theory is the same as that formulated on an elliptic fibration. It 
can be seen from \eqref{Retau}-\eqref{Imsigma}, that in our case, 
$\tau$, $\sigma$ and $\beta$ are not independent as in \cite{MMP}. We 
will give some examples of hyperelliptic curves satisfying these 
conditions in the following. Consequently, due to the split Jacobian 
property, there exists a surjection from these families of 
hyperelliptic curves fibered over the base manifold, to elliptic 
curves. However, we would like to emphasize, that this does not 
correspond to a fancy re-writing of an elliptic fibration. New types 
of singularities appear that do not exist in the elliptic case and the 
physical information around those singularities is also different.

After imposing these conditions, the symplectic matrix will depend on 
7 independent parameters.

Before presenting some particular examples, let us emphasize the 
existence of solutions to our proposal.

\begin{theorem}\label{prop10}
Let $Y$ be a hyperelliptic fibration with a period matrix $\Omega$. Then $\lambda=\det{\Omega}$ will transform as the axio-dilaton (\ref{axio-dilaton trans}) if and only if
\begin{equation*} \label{ph}
    \Omega(z) =  
       \begin{pmatrix}
        \frac{n_1}{k_2}  &  0  \\
        0   & \frac{n_2}{k_2}
    \end{pmatrix} 
     +\begin{pmatrix}
       \frac{k_3}{k_2} &  1  \\
        1  &  \frac{k_1}{k_2}
    \end{pmatrix}\beta(z)   \tag{{\bf ph}}
\end{equation*}
with
\begin{eqnarray*}
    n_1 &=& d'(B_{11}A_{22}-B_{21}A_{12})-b'(D_{11}C_{22}-D_{21}C_{12}) , \\
    n_2 &=& b'(D_{22}C_{11}-D_{12}C_{21})-d'(B_{22}A_{11}-B_{12}A_{21})
\end{eqnarray*}
integers obtained from (\ref{Retau}) and (\ref{Resigma}), respectively. 
\end{theorem}

\begin{proof} To start with, we have $\Omega \mapsto \frac{A\Omega+B}{C\Omega +D}$, and we must analyse the effect of this transformation on $\det(\Omega)$. 
Observe that, since \\

$    \det{(A\Omega + B)} =$
\begin{eqnarray*}
&=& \left( (A_{11}\left(\frac{n_1+k_3\beta}{k_2}\right)+A_{12}\beta +B_{11}) \right)\left( (A_{22}\left(\frac{n_2+k_1\beta}{k_2}\right)+A_{21}\beta +B_{22}) \right) \\ &-& \left( (A_{12}\left(\frac{n_2+k_1\beta}{k_2}\right)+A_{11}\beta +B_{12}) \right)\left( (A_{21}\left(\frac{n_1+k_3\beta}{k_2}\right)+A_{22}\beta +B_{21}) \right) \\&=& a\det{\Omega}+b,\\
\end{eqnarray*}

    $ \det{(C\Omega + D)}= $
     \begin{eqnarray*}
     &=& \left( (C_{11}\left(\frac{n_1+k_3\beta}{k_2}\right)+C_{12}\beta +D_{11}) \right)\left( (C_{22}\left(\frac{n_2+k_1\beta}{k_2}\right)+C_{21}\beta +D_{22}) \right) \\ &-& \left( (C_{12}\left(\frac{n_2+k_1\beta}{k_2}\right)+C_{11}\beta +D_{12}) \right)\left( (C_{21}\left(\frac{n_1+k_3\beta}{k_2}\right)+C_{22}\beta +D_{21}) \right) \\&=& c\det{\Omega}+d,
\end{eqnarray*}
we obtain
$$\lambda'=\det(\Omega')=\det\left( \frac{A\Omega+B}{C\Omega +D}\right)=\frac{a\lambda+b}{c\lambda+d},$$
where $a,b,c,d$ are integers given in equations \eqref{a}--\eqref{d}. 
\end{proof}       
Notice that, in the case when $n_1,n_2=0$, then $\tau=\sigma=\frac{k_3}{k_2}\beta$ and we have
$$
    b = \det{(B)}, \quad
    d = \det{(D)}.
$$

\subsection{Hyperelliptic curves satisfying the \ref{ph} condition}\label{phcond}
We present explicit  examples of hyperelliptic curves and hyperelliptic fibrations fulfilling Theorem \ref{prop10}.
\begin{example}\label{ex1}
 From theorem 2 from \cite{Schindler}, it can be seen that the period matrix of a hyperelliptic curve with equation
\begin{eqnarray*}
    y^2 = x^6-1
\end{eqnarray*}
is given by
\begin{eqnarray*} \Omega = i\frac{\sqrt{3}}{3}
\begin{pmatrix}
    2 & 1 \\
    1 & 2
\end{pmatrix}  .  
\end{eqnarray*}
This corresponds to a curve of type $D_{6}(0,1,0)$ \cite{Schindler}. It is easy to see that, 
as it is a pure imaginary matrix, the conditions (\ref{b}) and (\ref{d})  are trivially satisfied and $k_1=k_3=2k_2$ in (\ref{Imtau}) and (\ref{Imsigma}).

Applying the \textit{Mathematica} algorithm outlined in Appendix \ref{Appendix}, 
we can find symplectic matrices that ensure these conditions. Take for example, any matrix of the form
\begin{eqnarray}
    \begin{pmatrix}
        2 & -1-3k & -3 & -3 \\
        -1+k &1+k & 1 & 2 \\
        0 & k & 1 & 1 \\
        k &k & 1 & 2
    \end{pmatrix} \in Sp(4,\mathbb{Z}). \label{matrix1stex}
\end{eqnarray}
It can be verified that, in this particular case, the corresponding transformation on the axio-dilaton is given by
\begin{eqnarray*}
    \lambda \rightarrow \frac{(1+3k^2)\lambda-3}{-k^2\lambda +1}
\end{eqnarray*}
with $\lambda = -1$ and $k\in\mathbb{Z}$.

Another possible choice (not contained in \eqref{matrix1stex}) of symplectic matrix is
\begin{eqnarray}
    \begin{pmatrix}
        2 & -1-6k & -3k & -3k \\
        -1+2k &1+2k & k & 2k \\
        0 & 2 & 1 & 1 \\
        2 & 2 & 1 & 2
    \end{pmatrix} \in Sp(4,\mathbb{Z}) \label{matrix1stex2}.
\end{eqnarray}
In this case
\begin{eqnarray*}
    \lambda \rightarrow \frac{(1+12k^2)\lambda-3k^2}{-4\lambda +1}
\end{eqnarray*}
with $\lambda = -1$ and $k\in\mathbb{Z}$. 

Indeed, if we consider a hyperelliptic fibration whose corresponding period matrix can be written as
\begin{eqnarray*} \Omega(z) = if(z) 
\begin{pmatrix}
    2n & n \\
    n & 2n
\end{pmatrix}   
\end{eqnarray*}
with $n\in\mathbb{Z}$ and $z$ a 
coordinate on the base manifold, then the symplectic matrices 
(\ref{matrix1stex}) and (\ref{matrix1stex2}) will also ensure the $SL(2,\mathbb{Z})$ transformation for the axio-dilaton.
\end{example}

\begin{example}\label{ex2}
It can be seen that the period matrix of the Burnside curve
\begin{eqnarray*}
    y^2 = x(x^4-1)
\end{eqnarray*}
is given by
\begin{eqnarray*} \Omega = 
\begin{pmatrix}
   t & \frac{t}{2} \\
    \frac{t}{2} & t
\end{pmatrix}    
\end{eqnarray*}
with $t=\frac{1}{3}(-2+i2\sqrt{2})$ \footnote{We are grateful to Anita Rojas for clarification on this issue \cite{
Rojas2,Rojas,RojasRomero}.}. In this case, from   (\ref{Imsigma}) and (\ref{Imtau}) we obtain that  $k_1=k_3=2k_2$. Using the 
\textit{Mathematica} algorithm following the steps in Appendix \ref{Appendix},
we can ensure the existence of symplectic matrices satisfying these conditions and also (\ref{Retau}) and (\ref{Resigma}). 
Take for example 
\begin{eqnarray*}
    \begin{pmatrix}
        -2k & 1+k & k & 0 \\
        1+k(-3+2\tilde{k}) & 1+3k-(1+k)\tilde{k} & -k(-2+\tilde{k}) & k \\
        3-2\tilde{k} & \tilde{k} & -1+\tilde{k} & 1\\
        -2 & 1 & 1 & 0
    \end{pmatrix} \in Sp(4,\mathbb{Z}).
\end{eqnarray*}
where $k,\tilde{k}\in \mathbb{Z}$. It can be checked that, in these particular cases, the corresponding transformation on the axio-dilaton is given by
\begin{eqnarray*}
    \lambda \rightarrow \frac{(-1-3k^2)\lambda + k^2}{3\lambda -1}
\end{eqnarray*}
with $\lambda = -\frac{1}{3}(1+i2\sqrt{2})$. 
\end{example}

\begin{example}\label{ex3}
 In \cite{MMP} it was shown that the period matrix corresponding to the hyperelliptic curve
\begin{eqnarray*}
    y^2 = (x^3-z)^2 - g^2(x^3-m^3)^2
\end{eqnarray*}
 with $g,m$ constants, and $z$ coordinates on the base manifold $\mathbb{C}$, is given by
 \begin{eqnarray*}
     \Omega(z) = \frac{1}{2\pi i}\log{\left( 1-\frac{z}{m^3}\right)} \begin{pmatrix}
         4 & 2 \\ 2 & 4
     \end{pmatrix} .
 \end{eqnarray*}
 Observing the similarity of this period matrix to the one in example \ref{ex1},  
 it can be checked, using 
  the symplectic matrices (\ref{matrix1stex}) and (\ref{matrix1stex2}), that
 this example also satisfy the \eqref{ph} condition, so that the axio-dilaton transforms according to $SL(2,\mathbb{Z})$.
 
\end{example}

\begin{example}\label{ex4}
 If we consider the family of unimodular hyperelliptic curves with automorphism symmetry group $D_4$ we have  \cite{Rojas2,Rojas,RojasRomero}
\begin{eqnarray*}
    y^2 = x(x^2-1)(x-a)(x-a^{-1})
\end{eqnarray*}
with $a\in\mathbb{C}$. It can be seen that this family of curves contains the curves from the first and second examples as particular cases, i.e., if $a=i$ it corresponds to the Burnside curve. 

The period matrix related to this type of curve is given by
\begin{eqnarray*}
   \Omega = \begin{pmatrix}
        \tau & \beta \\
        \beta & \tau
    \end{pmatrix}
\end{eqnarray*}
with $|{\beta}|^2 = |\tau|^2-1$. In this case, we have that $k_1=k_3$ and using \textit{Mathematica}
 we find that there exist symplectic matrices, for example:
\begin{eqnarray*}
    \begin{pmatrix}
        0 & 1 & \mp k & -1 \\
        1 & \mp k & -1+k^2 & 0 \\
        0 & 0 & \pm k & 1 \\
        0 & 0 & 1 & 0
    \end{pmatrix} \in Sp(4,\mathbb{Z})
\end{eqnarray*}
with $k\in \mathbb{Z}$, such that the determinant of the period matrix transforms as the axio-dilaton.
 In this particular example, we have that there is a consistency condition given by $\tau=k\beta$, 
 which is satisfied for the first and second examples. 
 This condition is a particular case of Theorem \ref{prop10}, where $n_1=n_2=0$. 
 The corresponding transformation on the axio-dilaton is given by
\begin{eqnarray*}
    \lambda \rightarrow \frac{-k^2\lambda+(k^2-1)}{\lambda-1} 
\end{eqnarray*}
with $\displaystyle \lambda
=\frac{(k^2-1)(1-\beta^2k^2)}{1+\beta^2(k^2-1)}$  representing a physically admissible family of curves. \\
\end{example}

This family of hyperelliptic curves satisfies the split Jacobian condition from section (\ref{SJ}). 
Hence, there is a well-defined surjection from the hyperelliptic curve to an elliptic curve. 
Some implications of this issue will be addressed in the following section.

\begin{example} \label{ex5}
We may consider
\begin{eqnarray*}
    y^2 = x(x^2-1)(x-a(z))(x-a^{-1}(z))
\end{eqnarray*}
with $z\in\mathbb{C}$ local coordinates on the base manifold. On each point of the base manifold,
 there is an unimodular hyperelliptic curve with automorphism symmetry $D_4$. The period matrices related to this type of curve are given by
\begin{eqnarray*}
   \Omega(z) = \begin{pmatrix}
        \tau(z) & \beta(z) \\
        \beta(z) & \tau(z)
    \end{pmatrix}
\end{eqnarray*}
with $|{\beta}(z)|^2 = |\tau(z)|^2-1$. It can be checked, using the Mathematica algorithm based in Appendix \ref{Appendix},
that the symplectic matrices of the previous example are also valid.

We observe that when $a(z) = 1$ the curve becomes 
\begin{eqnarray*}
    y^2 = x(x-1)^3(x+1)
\end{eqnarray*}
therefore the hyperelliptic curve acquires a cusp singularity. 
Similarly,  when $a(z) = -1$ the curve becomes 
\begin{eqnarray*}
    y^2 = x(x+1)^3(x-1)
\end{eqnarray*}
and once again the hyperelliptic curve acquires a cusp singularity. 

If we consider the compact fibration inside a projective space, we can also make $a(z)=0$ and 
$a(z) = \infty$, but the 
family acquires extra singularities with higher multiplicities. 
For example for $a(z) =z$, we find  extra singularities of multiplicities greater than 3 at $z=0$ and at $z=\infty$.
Moreover, by analyzing the discriminant, it can be seen that the number of branes associated with $a(z)=0,1,-1,\infty$ 
coincide  \cite{Daflon, Lockhart}. For $a(z)=z$, we have to place a stack of 6 branes at each point.
\end{example}

\subsection{Hyperelliptic fibrations and monodromies}\label{monodromies}

F-theory on  elliptically fibered Calabi--Yau corresponds to type IIB string theory compactified on the base of the elliptic fibration. 
Working with an elliptically fibered Calabi--Yau ensures the geometrical interpretation of the axio-dilaton 
is described by  the Teichm\"uller parameters of  the elliptic fibers, 
and also guaranties the preservation of several properties like a proper amount of supersymmetry, cancellation of tadpoles, 
and backreaction effects.
The generalization to hyperelliptic fibrations is not evident if we want the fibration to be Calabi--Yau.

According to example \ref{aa2}, we have that $X=\Tot(\omega_Y)$ can be written as 
\[
\begin{tikzcd}[swap]
	\Tot(\omega_{\Sigma_p}) \arrow[]{r}[] {}
	& X \arrow[]{d}[right]{} \\  
		& \mathbb{P}^1 
\end{tikzcd}
\]
where the $Y$ is the hyperelliptic fibred surface  over $\mathbb{P}^1$ 
with fiber $\Sigma_p$.
 Consequently, in a small neighborhood of a smooth fiber, we may consider $\Tot(\omega_{\Sigma_p})$ also as
 a fiber of the $CY_4$. Moreover, by section \ref{phcond}, we know that there are families of hyperelliptic curves such that the determinant of the period matrix can be identified with the axio-dilaton together with the corresponding $SL(2,\mathbb{Z})$ transformation

 Let us emphasize that this interpretation differs from the one considered in \cite{Braun,Candelas1,Candelas2,MMP}, 
 where the axio-dilaton was given by $\tau$, that is, one of the entries of  the period matrix (\ref{Periodmatrix}). In our case, only such hyperelliptic curves 
 satisfying (\ref{ph}) will ensure the expected $SL(2,\mathbb{Z})$ transformation on the axio-dilaton $\lambda=\det\Omega$. We found that the hyperelliptic family of curves with automorphism symmetry $D_4$ \cite{Rojas2,Rojas,RojasRomero} satisfies this property. Moreover, this family of hyperelliptic curves also satisfies 
 the split 
 Jacobian condition, implying that there is a surjection from the hyperelliptic curve to an elliptic curve. 
 Considering  for instance the hyperelliptic fibration from example \ref{ex5}, we have that, on the points of the base manifold where the fiber is smooth,
  there is a surjective map from the genus two  fiber to an elliptic curve. However, singularities of hyperelliptic fibrations 
   are in general more complicated than singularities of elliptic fibrations.

In this work, we have mostly considered nodal fibers. These are singularities of the hyperelliptic fiber, but keeping $X=\Tot(\omega_Y)$ smooth. 
They correspond to type I in the Kodaira classification.
They are associated with the presence of $D7$ branes with worldvolume $\mathbb{R}^{1,7}$ located on the points of the base manifold whose
 fiber degenerates. These $D7$are related with the more general ($p,q$) 7-branes by $SL(2,\mathbb{Z})$ transformations.
  Nevertheless,  more complicated singularities appear in some cases.
For example, for the hyperelliptic fibration from example \ref{ex5}, we have that the singularities at $a=\pm 1$ are cusps. Let us note that, in example \ref{ex5}, fibers associated with purely imaginary period matrices are related to infinite distances in the moduli space.

In F-theory, moving around a puncture, the axio-dilaton undergoes a monodromy in $SL(2,\mathbb{Z})$. In \cite{Braun,Candelas1,Candelas2,MMP}, the monodromy of the period matrix around  singular  fibers of the $K3$  is contained in the U-duality group, $SO(2,n,\mathbb{Z})$. For $n=3$, we have that the U-duality group is isomorphic to $Sp(4,\mathbb{Z})$, the modular group of a genus 2 curve. In our case, the $SL(2,\mathbb{Z})$ transformation of the axio-dilaton is inherited from $Sp(4,\mathbb{Z})$. By doing this, we can obtain parabolic monodromies associated with the ($p,q$) 7-branes, but also monodromies associated with the other conjugacy classes of $SL(2,\mathbb{Z})$. 

We may consider stacks of ($p,q$) 7-branes having monodromy groups that are not necessarily parabolic but,
alternatively, are identified with ADE groups  \cite{deWolfe,deWolfe2}. In particular, stacks of 
  ($p,q$) 7-branes associated with $SO(8)$, $E_6$, $E_7$ and $E_8$ correspond to elliptic monodromies given by $\mathbb{Z}_2$, $\mathbb{Z}_3$, $\mathbb{Z}_4$ and $\mathbb{Z}_6$, respectively. 

Brane solutions related with elliptic and hyperbolic conjugacy classes were discussed in \cite{Bergshoeff7, Bergshoeff8,Hartong} from type IIB and F-theory perspectives. In \cite{Bergshoeff7} it is shown that branes with elliptic monodromy play a crucial role in type IIB while branes with hyperbolic monodromies never arise. They proposed that these Q7 may be interpreted as stacks of ($p,q$) 7-branes. However, in \cite{Bergshoeff9} it is claimed that a consistent action invariant under $\kappa$-symmetry can only be built for the parabolic case. We leave the discussion of  elliptic monodromies and stacks of ($p,q$) 7-branes in this context to future work.

In  compactifications of F-theory with fluxes, as well as in  non-supersymmetric scenarios, the Calabi--Yau condition is sometimes not imposed. 
We may also consider hyperelliptic fibrations in such cases, keeping with the identification of the axio-dilaton to the determinant of the 
period matrix as we have shown in the previous section. This approach was not considered in this paper,  we leave it for future work.


\section{Applications to constructions in F-theory II} \label{down}

\subsection{The hyperelliptic fibration occurs downstairs}\label{downstairs}
 In section \ref{totomega}, we
showed how to  construct a  Calabi--Yau threefold $X=\Tot(\omega_Y)$, realizing  
any  even value of $c_2(X)$, which
has the novelty of containing  a hyperelliptically fibered surface $Y$.
In this section we consider the standard fibrations 
 \[
\begin{tikzcd}[swap]
	T^2 \arrow[]{r}[] {}
	& CY_4 \arrow[]{d}[right]{} \\  
		& CY_3
\end{tikzcd}
\]
where the $CY_4$ is elliptically fibered over the $CY_3$, and we discuss some 
consequences of the presence of the hyperelliptic fibration inside the $CY_3$.
  In general, our fibrations will not be  locally trivial; 
that is, the fiber $T^2$ may degenerate, becoming singular, yet having genus 1.

F-theory on an elliptically fibered $CY_4$  corresponds, in our construction, to type IIB superstring theory on a $CY_3$. 
The type IIB axio-dilaton that models the string coupling corresponds to the Teichm\"uller parameter of the elliptic fibration over the $CY_3$. 
 One might inquire what role such $CY_3$'s  with arbitrarily high (or arbitrarily low) values of even $c_2$ 
 play in the physics of the specific types of $CY_4$'s in the context of F/String theories.

To explore applications of these $CY_4$'s let us recall the following:
\newline
Type IIB string phenomenological constructions often 
require D3-branes and D7-branes.  The type II D3-branes fill the spacetime dimensions, and they are points within the extra dimensions. They do not exert any backreaction on the space and are invariant under the $SL(2,\mathbb{Z})$. Consequently, they do not receive any correction when uplifted from type IIB theory to F-theory. This is significant because the same D3-brane content is present in F-theory. However, this is not the case for the D7-branes, which generically require the backreaction to be taken into account, and at the level of F-theory  are generalized to $(p,q)$ 7-branes. 

Furthermore, consistency of either of these theories requires that anomaly cancellation 
be guaranteed. In the absence of 3-form fluxes in type IIB constructions, or respectively, 4-form fluxes in F-theory, the F-theory anomaly cancellation implies the type IIB anomaly cancellation. Anomalies occur when a classical symmetry is not preserved at the quantum level. This is usually manifested in the non-invariance of the path integral under the symmetry transformation. Depending on the type of symmetry breaking, there are different types of anomalies in quantum theories. 
 In type II string theory, there are the gauge, the gravitational, and the mixed anomalies. 
  It is out of the scope of the present article to give a detailed study of anomaly cancellation, we present a brief account, and refer the interested reader to 
   \cite{AE, Weigand3,Blumenhagen, Yeuk,  deWolfe,  Inaki,  Ibanez, W-Taylor, Uranga2}. 
 \newline

Often, in type IIB string theory, cancellation of anomalies is obtained by imposing 
what is known as \textit{tadpole cancellation}
\footnote{For the case of mixed anomalies on top of it, one must also consider the Green--Schwarz mechanism.} \cite{Aldazabal-tadpole}.  
Following the Gauss law, the total charge on a compact manifold 
must vanish—this is usually known as the RR tadpole cancellation. 
In supersymmetric theories,  RR tadpole cancellation also implies that  NS-NS tadpoles are  cancelled and the theory is consistent. 
This cancellation must hold for any of the charges present in the construction. 
To this end, cancellation of all tadpoles, in particular those associated with the D3-brane, 
D5-brane, and D7-brane RR charges, is required \cite{Blumenhagen}. \newline

In type IIB theory the Dp-brane action, $S_{Dp}= S_{DBI}+S_{WZ}$, is described by 
the following integrals over the worldvolume (WV):

\begin{equation}
\begin{aligned}
&S_{DBI} = - T_p\int_{WV}
d^{(p+1)}\xi e^{-\phi}\sqrt{ \det (g + \mathcal{F})}.\\
&S_{WZ}=
2\pi T_p\int_{WV}\sum_{k}C_{2k}\Tr\left[e^{\frac{1}{2\pi}\mathcal{F}}\right]\sqrt{\frac{\widehat{A}(T)}{\widehat{A}(N)}}\label{WZ}
\end{aligned}
\end{equation}
where  $T_p$ represents the tension of the Dp-brane that we will set equal to 1,
 $\phi$ is the dilaton, $g$ is the determinant of the induced metric over the 
 worldvolume of the Dp-brane, and $\mathcal{F}$ is defined as
\begin{equation}
    \mathcal{F}=F+i^*B,
\end{equation}
where $i$ is a map from the world volume to the target space. For simplicity, we will consider the case $B=0$.
The terms $C_{2k}$ are  RR forms present in type IIB theory, $\widehat{A}$ 
is the roof-genus, 
  and in \eqref{WZ} the letters $N$ and $T$ denotes the normal and tangent bundle, respectively.

We  particularize the matter content to the target space considered,
 where $X=\Tot(\omega_Y)$ is a $CY_3$ with nontrivial $c_2$ and the surface $ Y \to \mathbb P^1$ is hyperelliptically fibered,
 and  consider constructions with  D3-branes and D7-branes.

 D7-tadpoles must be canceled.  The D7-tadpole associated with the 
 $D7$ charge can be canceled in F-theory by a proper combination of $(p,q)$ 7-branes 
 that coalesce to form a singularity that cancels the $D7$charge \cite{deWolfe}. 
The D7-tadpole cancellation in F-theory is codified in the topology
 of the elliptically fibered $CY_4$. 

The wrapping of the $D7$ branes along those 2-cycles labelled with and index $a$ contained in $Y$  forming a basis of $H_2(Y)$ 
 can also induce a D5-brane charge associated with a WZ term
\begin{equation}
    S_{D7_{5}}=\int \sum_aC_6\wedge \mathcal{F}_a.
\end{equation}
This generates a D5-tadpole.
Typically, in type IIB theory, 
 orientifold planes are introduced \cite{Dabholkar} to cancel these $D_5$-tadpoles. However, for the 
  cases we consider here, we may assume 
 that the $D7_{5}$ charge contribution vanishes. Therefore, for the particular 
 case where the $D7$branes wrap both 2-cycles   simultaneously (a fiber and the base of $Y$), 
 equality (\ref{D5}) implies that the D5-brane tadpoles here may cancel out without the need of 
 introducing orientifold planes. The topology of our genus 2 hyperelliptic fibration $Y$ facilitates the cancellation since,
\begin{equation}\label{D5}
  (\chi(\Sigma_2)+\chi(\mathbb{P}^1))=0,
\end{equation} 
observing that, for a curve  $\chi$ agrees with the first Chern number. \\

D7-branes may also induce a D3-brane charge. The contribution to the action of the term associated with the roof-genus is \cite{Plauschimm},
\begin{equation}
\sqrt{\frac{\widehat{\mathcal{A}}(T)}{\widehat{\mathcal{A}}(N)}}=
\large(1+\frac{1}{96}(\frac{l_s^2}{2\pi})^2 \Tr(R^2)+\dots\large)\wedge (1+\frac{l_s^4}{24}c_2(Y)+\dots)
\end{equation}
where $\Gamma_{D7}$ is the holomorphic 4-cycle wrapped by the brane ($Y$ in our case) and $R$ is a curvature 2-form of the tangent bundle of $\mathbb{R}^{1,3}$. Consequently, writing  $N_{D7}$ to represent the number of $D7$branes considered in a particular model, from \eqref{WZ} we obtain the following term
 \begin{equation}
    S_{WZ}\supset T_7\int_{\mathbb{R}^{3,1}}C_4\wedge\int_{Y}\left(\Tr \mathcal{F}\wedge \mathcal{F}+l_s^4 N_{D7}\frac{c_2(Y)}{24}\right)\label{WZP}. 
 \end{equation}
 The induced charge has two different physical origins: this occurs whenever
  the $D7$that fills out the 4D spacetime wraps any four cycles contained in the $CY_3$ generating $c_2(Y)$, and  
 the  other term is the contribution of magnetized D7's associated to $\Tr \mathcal{F}\wedge \mathcal{F}$.
 \\

 Cancellation of D3-tadpoles may be ensured by inducing a  D3-charge 
 that compensates the one induced by the curvature of the $CY_4$ \cite{GKP}. 
The $D_3$ brane charge 
generated by the matter content may have three different origins or a combination of them: either it is associated with the presence of D3-branes, or it is induced by magnetized D7-branes, or produced by 3-form fluxes, or else a combination of those. In \cite{Dasgupta2} it was shown that due to the Chern--Simons coupling 
when the $D7$branes are compactified on a 4-cycle $Y$, the D3-branes can be understood as instantons over the $D7$branes.
Let us consider the contribution of D7-branes wrapping the holomorphic divisor $Y$ of $X$ 
without the presence of gauge fluxes. One can observe that it generates an induced $D3$charge on the D7-brane of geometric origin:
\begin{equation*}
Q_{\mathrm{geometric}}^{D7_3}= N_{D7}\frac{\chi(Y)}{24}, 
\end{equation*}
with $\chi(Y)= c_2(Y)$
\footnote{In general, the Euler characteristic must be corrected to 
$\chi_0(\Gamma_{D7})=\chi(\Gamma_{D7}) -n_{pp}$ 
where $n_{pp}$ is the number of singularities (pinched points)
 \cite{AE}.}
 in our case. 
 This term contributes, in general, to the Euler characteristic of the Calabi Yau fourfold \cite{Blumenhagen}. In the set-up considered in this paper, the D7-branes intersect the fibers $\omega_Y$ at points,  wrapping the four-cycle associated to $Y$, i.e., the D7-branes may be seen as a multivalued section of the canonical bundle of the hyperelliptic fibered surface $Y$.  
  This scenario is only possible for the restricted values of the moduli where the D7-backreaction does not need to be considered. This region corresponds to $|z-z_0|<< \lambda$, where $z_0$ is the location of the brane.\\
  
By Lemma \ref{c2}, equality $c_2(X)=-c_1^2(Y)+c_2(Y)$
 applies,  accounting for the contribution of $c_1^2(Y)$ towards $c_2(X)$
 brought about by the $A$-roof genus of the normal bundle $N$,  thus  reflecting the nontriviallity of  the embedding of $Y$ in $X$.
As proved in Theorem \ref{ue1}, when $Y$ is a smooth projective surface, we have that $c_2(X)$ is even. 
This implies that the number of times that D7's can be wrapped becomes restricted  
for a given $X$. Equivalently, a particular
$D7$wrapping can only be turned on  particular $CY_3$’s. The $D7$wrapping generates a D3
charge that contributes to that one generated by the $CY_4$.\\

There  exists the possibility of having nontrivial gauge fluxes on the 7-branes 
that will also induce a contribution to the D3-charge. Magnetized D-branes have been 
used for type IIB phenomenological constructions in a number of papers
 in order to generate D-terms via anomalous U(1)'s \cite{Ibanez,JL}  to guarantee chirality \cite{Cascales, Cicoli, Ishibashi},
they have been used in  model building \cite{Blumenhagen,font,Weigand2}, uplift to De Sitter \cite{de Sitter, Susha}, 
and even stabilize moduli \cite{GM}. 
A magnetized $D7$brane is a D7-brane that 
contains a DBI flux associated with a non-trivial $U(1)$ line bundle $l$ with a non-vanishing first Chern class $c_1$  given by 
\begin{equation*}
\int_{\Sigma}
F_2= k\in \mathbb{Z}-\{0\}=c_1(l)
\end{equation*}
where $\Sigma$ is a 2-cycle contained in the $CY_3$.
The gauge flux condition can
be defined on $\Sigma_a$ of arbitrary genus \cite{mr}, hence also for 
hyperelliptic curves. 
We have that the contribution to the induced charge due to gauge fluxes is given by
\begin{equation}
\label{gauge}
N_{\mathrm{gauge}}=-\frac{1}{2}\int_{Y}\mathcal{F}\wedge\mathcal{F}.
\end{equation} 
\newline
When gauge fluxes of the DBI field strength are turned on a brane, consistency also requires imposing a 
quantization condition to cancel out
 the Freed--Witten anomaly. Assuming for simplicity that the pull-back of the NSNS B-field is zero, this is obtained by choosing  \cite{FW,W}
\begin{equation}
c_1(l_a)+\frac{c_1(Y)}{2}\in H^{2}(X,Z).
\end{equation}
F-theory on an  elliptically fibered compact $CY_4$ has a tadpole associated 
with the 4-form $C_4$,  given by $N_{\chi_4}=-\frac{\chi(CY_4)}{24}$ \cite{Green}. 
 Hence, in such a case, the tadpole cancellation condition is 
\begin{equation}\label{tad2}
N_{\chi_4}+N_{D_3}+N_{\mathrm{gauge}}+N_{\mathrm{fluxes}}=0.
\end{equation} 
We did not include 3 or 4-form fluxes acting on the CY construction in this text,
leaving it for future work. 
For the  D-brane content considered here, supersymmetry is preserved, 
and tadpole cancellation occurs when (\ref{tad2}) is satisfied. 
As we have signaled, in the absence of $D3$ branes, 
the induced charge due to gauge fluxes $N_{\mathrm{gauge}}$  will be necessary to cancel D3-tadpoles. 
In order to introduce matter content, singularities are required, as usual in F-theory constructions. 
A relevant property for constructions in  F-theory that support GUT models with enough positive 
$D3$ charges and branes while also having enough nonvanishing background fluxes for  stabilizing
the $D7$ moduli perturbatively, is to have $\chi(CY_4)$ large enough \cite{Blumenhagen}. 
See also \cite{Bena} about the relevance of this requirement. An interesting property of our construction is that 
 $c_2(X)$ can be arbitrarily large (or also arbitrarily small). 
Hence, the $CY_4$ in our constructions may also admit arbitrary large Euler number, pending on 
adding more singular fibers.
Then, $N_{\chi_4}$ will contribute negatively towards the $D3$-tadpole cancellation condition. 
 There are a few technical issues to be considered when adapting the use of $N_{\chi_4}$ to our models, 
 because of  noncompactness. Even though there is an easy to fix of this issue, by observing that each of
 our $CY_4$'s has the homotopy type of a compact threefold, there are many interesting features to 
 be discussed, so 
we leave a more detailed analysis of concrete model building with particular types of singularities 
in the fibers of our $CY_4$'s also for future work.\\

Summarizing, we have shown that the  constructions of Calabi--Yau manifolds
 made here allow for  many choices 
of wrapping of  $D7$-branes in type IIB, with induced $D3$ charges, and
 allowing for the possibility of cancellation of $D5$-brane tadpoles without introducing orientifold planes. 
Other examples that explore the occurrence of hyperelliptic curves for phenomenological purposes in string theory 
appear in \cite{Adams,Kachru,Marino,MMP,r,Saltman}, and references therein.


\section{Conclusions}\label{conclusion} We have presented 
a mathematical construction of  Calabi--Yau manifolds, 
 containing hyperelliptic fibrations with arbitrary even values of 
the second Chern class.  
 We analyze the construction of F-theory over a $CY_4$, on this type of background with hyperelliptic curves in two different situations: \\

In the first case, we formulate F-theory over a noncompact $CY_4$, containing  
a consistent  $CY_3$, which  is itself  $\Tot(\omega_Y)$,
where $Y$ is a hyperelliptic fibration over $\mathbb{P}^1$. 
 This is an extension of F-theory. The role of the axio-dilaton, which parametrizes the coupling constant in 
 Type IIB string theory, is associated with the moduli of the hyperelliptic fiber.
 We propose that the axio-dilaton should correspond to the determinant of the period matrix. 
 We have found that the family of hyperelliptic curves with automorphism group $D_4$ satisfies this condition. 
 Outside the singular points, there is a surjection from the
hyperelliptic fiber to an elliptic fiber due to the split Jacobian condition. However, the singularities of hyperelliptic curves may in general be
 different from those of elliptic curves. The $SL(2,\mathbb{Z})$ transformation is inherited from the modular group of the genus 2 curve. 
For the hyperelliptic fibration considered in the example \ref{ex5}, 
 we found that there are cusp singularities in points of the base manifold where the fiber degenerates.\\

In the second case, the $CY_4$ is a fibered over the $CY_3$ with 
elliptic fibers. The presence of hyperelliptic curves inside the $CY_3$ determines 
 the embedding of branes that satisfy the tadpole cancellation conditions.
 In our construction,  the possibility of having large values of $c_2(X)$ implies
  that the $CY_4$ may also admit  large (local) Euler characteristic, 
whose  contribution
 facilitates the $D3$tadpole cancellation condition in models containing background fluxes.
 This is a desired property since GUT models with enough positive $D3$charges 
 and branes may need a large number of fluxes that stabilize D7-brane moduli to cancel tadpoles.
 \\
 
If the requirement that the background be  Calabi--Yau is relaxed, either by introducing fluxes that imply loss of Ricci flatness or by compactifying on other 
types of manifolds,  hyperelliptic curves could be included without the need to introduce their nontrivial canonical bundles; hence, in such cases,
the corresponding background where type IIB string theory is formulated can be more easily identified. We leave this study for future work.

 \section{Acknowledgments}

We thank A. Uranga, A. Rojas, and K. Ray for helpful discussions. 

Ballico  is a member of MUR and GNSAGA of INdAM (Italy).

Gasparim is a Senior Associate of the Abdus Salam International Centre for Theoretical Physics, a member of 
MathAmSud 24-MATH-12 , and was supported by MUR-PRIN 2017 
for visiting Trento (Italy).

Garcia del Moral is grateful to the Physics Dept. of Univ. Antofagasta (Chile) for their kind invitation, 
where part of this work was carry out. She is partially supported by 
PID2021-125700NB-C21 MCI grant (Spain). 

 Las Heras is partially supported by ANID 
posdoctorado Chile/2022-74220044. He is also grateful for the support of Project 
PID2021-123017NB-I00, funded by MCIN/AEI/10.13039/501100011033.

\appendix
\section{Sketch of a \textit{Mathematica} algorithm}\label{Appendix}
In this section, we will enumerate the basic steps underlying  the \textit{Mathematica} algorithm used to verify the examples in section \ref{AX}.
\begin{enumerate}
    \item We consider a matrix $4\times4$ given by $$M=\begin{pmatrix}
        A & B \\
        C & D
    \end{pmatrix}$$
    where $A,B,D,C$ are $2\times 2$ matrices.
    \item We impose $\det{(M)}=1$ and $\det{(A)}d-\det{(C)}b=1$. In the examples considered in this paper, $b=\det{(B)}$ and $d=\det{(D)}$ which corresponds to $n_1,n_2=0$. We are left with 14 (out of 16) free entries of the matrix $M$.
    \item We impose the matrix $M$ to be symplectic, that is
    $$M\begin{pmatrix}
        0 & \mathbb{I} \\
        -\mathbb{I} & 0
    \end{pmatrix}M^T=\begin{pmatrix}
        0 & \mathbb{I} \\
        -\mathbb{I} & 0
    \end{pmatrix}$$
    We are left with 9 (out of 14) entries of the matrix $M$. \\
    At this point, we have a generic matrix of $Sp(4,\mathbb{Z})$ such that $\det{(A)}d-\det{(C)}b=1$
    \item From now on, we must consider a specific hyperelliptic curve (or fibration) with its corresponding period matrix. 
    
    In examples (\ref{ex1}), (\ref{ex2}) and (\ref{ex3}), we have to impose, $k_1=\alpha k_2$ and $k_3=\gamma k_2$ with $\alpha,\gamma \in\mathbb{Z}$ that depends on each case.

    In examples (\ref{ex4}) and (\ref{ex5}), we have to impose, $k_1=k_3$ and $\vert\beta\vert^2=\vert \tau \vert^2-1$.

    In all cases, we are left with 7 (out of 9) entries of the symplectic matrix $M$.
    \item Finally, we choose some integer values for the 7 parameters left and verify that the resultant symplectic matrix is such that, the determinant of the period matrix transform according to $SL(2,\mathbb{Z})$.
\end{enumerate}

\end{document}